\documentclass[11pt]{article}
\usepackage{amsfonts}
\usepackage{indentfirst}
\usepackage{amsmath}
\usepackage{latexsym,amsmath,amssymb,paralist,subfigure,graphicx}
\usepackage{enumerate}
\usepackage{multirow}
\usepackage{comment}
\usepackage[compress]{cite}
\usepackage{algorithm}
\usepackage{algorithmic}
\usepackage[title]{appendix}
\usepackage{xcolor}
\usepackage{graphicx}
\usepackage{epstopdf}
\usepackage{bm}
\usepackage{hyperref}
\usepackage[normalem]{ulem}
\usepackage{bbm}
\usepackage{mathrsfs}
\usepackage{booktabs}
\usepackage{slashed}

\numberwithin{equation}{section}
\newtheorem{theorem}{Theorem}[section]
\newtheorem{lemma}{Lemma}[section]
\newtheorem{assumption}{Assumption}[section]
\newtheorem{proposition}{Proposition}[section]

\newtheorem{example}{Example}[section]
\newtheorem{remark}{Remark}[section]
\newtheorem{definition}{Definition}[section]
\newenvironment{proof}{{\noindent\it Proof:\quad}}{\hfill$\square$\par}
\newcommand{\qed}{\nobreak \ifvmode \relax \else
      \ifdim\lastskip<1.5em \hskip-\lastskip
      \hskip1.5em plus0em minus0.5em \fi \nobreak
      \vrule height0.75em width0.5em depth0.25em\fi}

\def\XXint#1#2#3{{\setbox0=\hbox{$#1{#2#3}{\int}$}
\vcenter{\hbox{$#2#3$}}\kern-.51\wd0}}

\begin{document}

\title{Numerical methods for  two-dimensional
G-heat equation}
\author{Ziting Pei, \ \ Xingye Yue\thanks{%
The corresponding author. Email:xyyue@suda.edu.cn (X. Yue).}, \ \ Xiaotao Zheng}
\date{\today }
\maketitle
\begin{abstract}
The G-expectation is a sublinear expectation. It is an important tool for pricing financial products and
managing risk thanks to its ability to deal with model uncertainty. The problem is how to
efficiently quantify it since the commonly used Monte Carlo method does not work. 
Fortunately, the expectation of a G-normal random variable can be linked to the viscosity solution of a fully nonlinear G-heat equation. In this paper, we propose a novel numerical scheme for the two-dimensional G-heat equation and pay more attention to the case that there exists uncertainty on the correlationship, especially to the case that the correlationship ranges from negative to positive. The scheme is monotonic, stable, and convergent.  The numerical tests show that the scheme is highly efficient.

\end{abstract}

\section{Introduction}

 In 2006, Peng\textsuperscript{\cite{p06}} introduced the so-called  G expectation to treat problems with model uncertainty. It has developed rapidly in order to respond to the
increasing demand for robust quantitative analysis and risk management.
Moreover, according to Peng%
\textsuperscript{\cite{p10}}, G-expectation is a coherent risk measure that
satisfies all the axioms proposed in Artzner et al.\textsuperscript{\cite{AD}}. 
It is difficult to determine G-expectation by Monte Carlo sampling, since the distribution of the random variable is uncertain. However, the G-expectation is related to a fully nonlinear G-heat equation\textsuperscript{\cite{p10}}.  One can quantify the G-expectation by numerically solving the G-heat equation.

The work of Barles and Souganidis\textsuperscript{\cite{G1991Convergence}} provided a theoretical foundation for fully nonlinear second-order equations that the numerical solution of a consistent and monotonic scheme converges to the viscosity solution of the original equation. For the one-dimensional case, the G-heat equation appears early as an option pricing model with volatility uncertainty, Pooley et al.\textsuperscript{\cite{pooley}} developed numerical
algorithms and discussed their convergence properties. For the multi-dimensional G-heat equations, it is non-trivial to construct a monotone scheme to ensure its convergence to the viscosity solution (Barles and Souganidis
\textsuperscript{\cite{G1991Convergence}}, Barles et al.\textsuperscript{%
\cite{1995CONVERGENCE}}). In this paper, we take the two-dimensional G-heat equation as an example, while the three-dimensional case can be analyzed analogously.

The two-dimensional G-heat equation also appears early as a two-factor uncertain volatility model. Pooley et al.\textsuperscript{\cite{pooleyb}} numerically solved the equation; however, the scheme was not guaranteed to be monotone. The main difficulty in constructing
monotone schemes is to treat the cross-derivative term. When the sign of correlationship is determined, a compact seven-point stencil ({\O }ksendal and Sulem%
\textsuperscript{\cite{KSENDAL}}, Clift and Forsyth\textsuperscript{%
\cite{clift2008numerical}}) that relies on the sign of the correlationship is employed for the discretization of the cross-derivative.
 To ensure monotonicity for problems with the cross-derivative, Bonnans and Zidani%
\textsuperscript{\cite{Bonnans2003}}, Debrabant and Jakobsen%
\textsuperscript{\cite{Debrabant2013}} focused on explicit wide stencil
schemes, while Ma and Forsyth\textsuperscript{\cite{MF}} proposed an implicit numerical scheme that combines the use of a fixed-point stencil and a wide stencil based on a local coordinate rotation.

However, so far, there has been no discussion that takes into account the situation in which the sign of the correlationship is uncertain. The approximation for the second-order cross-derivative plays a key role in obtaining a monotonic scheme. The selection of the seven-point stencil depends on the sign of the correlationship, while the sign of the correlationship, in turn, depends on the selected seven-point stencil. To break this cycle of dilemmas, we develop a novel implicit numerical scheme that is stable, consistent, and monotone. Thus, our numerical
scheme guarantees the convergence to the viscosity solution. 

The organization of this paper is as follows. In Section \ref{sec:2}, we
review the basic concepts of the G-expectation and the G-heat equation. In Section \ref%
{sec:3},  we
develop an implicit numerical scheme to solve the general
two-dimensional G-heat equation for the case where the correlationship varies from negative to positive. In Section \ref{sec:4}, we show the monotonicity of the scheme, which guarantees convergence to the viscosity solution. We present an estimate of the convergence rate. In particular, we show that the
non-linear iteration at each timestep is always convergent. In Section \ref%
{sec:5}, we validate the efficiency of our numerical scheme through
numerical examples. Finally, we provide some conclusions in Section \ref%
{section4}.

\section{Background}

{\label{sec:2}} 
In this section, we recall some basic knowledge about
Peng's G-stochastic calculus. The readers are referred to \textrm{\cite{p10}} for
more information.

\begin{definition}
The G-expectation $\mathbb{E}$ is a sublinear expectation that is a functional $\mathbb{E}$: ${\mathcal{H\mapsto}}\mathbb{R}$ satisfying

(a) Monotonicity: If $X\geq Y$, then $\mathbb{E[}X\mathbb{]}\geq
\mathbb{E[}Y\mathbb{]}$.

(b) Constant preserving: $\mathbb{E[}c\mathbb{]=}c,\forall c\in\mathbb{R}$.

(c) Sub-additivity: $\mathbb{E[}X+Y\mathbb{]}\leq\mathbb{E[}X\mathbb{]}%
+\mathbb{E[}Y\mathbb{]}$.

(d) Positive homogeneity: $\mathbb{E[\lambda}X\mathbb{]=\lambda E[}%
X\mathbb{]},\forall\lambda\geq0$.
\end{definition}

\begin{definition}
Let $X_{1}$ and $X_{2}$ be two $d$-dimensional random vectors defined on the
sublinear expectation spaces $(\Omega,\mathcal{H},\mathbb{E})$. They are
called identically distributed, denoted by $X_{1}\overset{d}{=}X_{2}$, if%
\[
\mathbb{E}[\varphi(X_{1})]=\mathbb{E}[\varphi(X_{2})],\ \forall\varphi\in
C_{l.Lip}(\mathbb{R}^{d}).
\]

\end{definition}

\begin{definition}
In a sublinear expectation space $(\Omega,\mathcal{H},\mathbb{E})$, a random
vector $Y\in{\mathcal{H}}^{d}$ is said to be independent of another random
vector $X\in{\mathcal{H}}^{d}$ under $\mathbb{E}$ if for each test function
$\varphi\in C_{l.Lip}(\mathbb{R}^{2d})$ we have%
\[
\mathbb{E}[\varphi(X,Y)]=\mathbb{E}\left[  \mathbb{E}[\varphi(x,Y)]_{x=X}%
\right]  .
\]

\end{definition}

\begin{definition}
($G$\textbf{-normal distribution)}. A $d$-dimensional random vector
$X=(X_{1},...,X_{d})$ in a sublinear expectation space $(\Omega,\mathcal{H}%
,\mathbb{E})$ is called G-normal distributed if for each $a,b>0$ we have%
\[
aX+b\overline{X}\,\overset{d}{=}\sqrt{a^{2}+b^{2}}X
\]
where $\overline{X}$ is an independent copy of $X$.
\end{definition}

Given a G-normal random variable $X\in \mathbb{R}^{d}$, we want to quantify its G-expectation $\mathbb{E}[\varphi (X)]$ for some application $\varphi$. Due to the uncertainty of covariance, we do not know how to obtain the samples for $X$, so the common-used  Monte Carlo simulation does not work for the G-expectation.   Let $u(t,x)=\mathbb{E}[\varphi (x+\sqrt{t%
}X)]$, 
 Peng\textsuperscript{\cite{p10}}shows that 
$u$ is the viscosity solution of the following so-called G-equation%
\begin{equation}
\begin{cases}
\partial _{t}u-G(D^{2}u)=0, & (t,x)\in (0,\infty )\times \mathbb{R}^{d}, \\
u(0,x)=\varphi (x), &
\end{cases}
\label{G_heat}
\end{equation}%
where $D^{2}u$ is the Hessian matrix of $u$, and
\begin{equation}
\begin{aligned} G(A) &:= \frac{1}{2}{\mathbb{E}}[ \langle AX,X \rangle], \\
&= \frac{1}{2}\sup_{Q \in \Theta} \text{Tr}[AQ], \end{aligned}
\label{G-generator1}
\end{equation}%
where ${A\in }\mathbb{S}(d)${, }$\mathbb{S}(d)$ denotes the space of $%
d\times d$ symmetric matrices, $\Theta $ represents the set of all possible
symmetric matrices, which is a given bounded, closed, nonnegative-definite and
convex subset of $\mathbb{R}^{d\times d}$.   
If $\Theta $ is a singleton: $\Theta =\{Q\}$,
then $X$ is a classical zero mean normal distributed with covariance $Q$. In general, $\Theta $ characterizes the covariance uncertainty of $X$,
and
\begin{equation}
\begin{aligned} Q &= \left( \begin{array}{cccc} {\sigma _1^2} & {b _{12}} &
\cdots & {b_{1d}} \\ {b_{21}} & {\sigma _2^2} & \cdots & {b _{2d}} \\ \vdots
& \vdots & \ddots & \vdots \\ {b_{d1}} & {b _{d2}} & \cdots & {\sigma _d^2}
\end{array} \right), \end{aligned}  \label{covariance_matrix}
\end{equation}%
is a symmetric nonnegative definite matrix and $b_{i,j}=b_{j,i}$. 

So, if we can solve G-equation \ref{G_heat}, we get $\mathbb{E}[\varphi (X)]=u(1,0)$.


In this paper, we will only consider two-dimensional problems, that is, $X=(X_1, X_2)$.   \eqref{G_heat} can be rewritten
as%
\begin{equation}
\left\{
\begin{array}{l}
u_{t}-\sup\limits_{Q\in \Theta }(\frac{\sigma _{1}^{2}}{2}u_{xx}+\frac{%
\sigma _{2}^{2}}{2}u_{yy}+b_{12}u_{xy})=0,(t,x,y)\in
(0,\infty )\times \mathbb{R}\times \mathbb{R}, \\
u(0,x,y)=\varphi (x,y),%
\end{array}%
\right.  \label{two-dimensional}
\end{equation}%
where the uncertainty of $Q$ can be identified by testing the proper symmetric matrices $A$
in equation \eqref{G-generator1}. Choosing, respectively, in equation (\ref%
{G-generator1}),
\begin{equation}  \label{A1}
\begin{aligned} A_1 = \left( {\begin{array}{*{20}{c}} 1&0\\ 0&0 \end{array}}
\right),A_2 = \left( {\begin{array}{*{20}{c}} { - 1}&0\\ 0&0 \end{array}}
\right), \end{aligned}
\end{equation}
we can derive the uncertainty in the variance of $X_1$,
\begin{equation}
\left\{
\begin{array}{l}
\underset{\sigma _{1}\in \Gamma _{1}}{\sup }\sigma _{1}^{2} = \mathbb{E}\left[ X_{1}^{2}\right],\\
\underset{\sigma _{1}\in \Gamma _{1}}{\inf }\sigma _{1}^{2}=-\mathbb{E}\left[-X_{1}^{2}\right].\end{array}\right.  \label{2d_uncertain_volatility}
\end{equation}It is evident that $\sigma _{1}^{2}\in \Gamma _{1} \stackrel{\triangle}{=} [-\mathbb{E}{[-X_{1}^{2}]},\mathbb{E}{[X_{1}^{2}]}]$.

Similarly, by choosing matrix $A$ in equation \eqref{G-generator1} as:
\begin{equation}
\begin{aligned} A_3 = \left( {\begin{array}{*{20}{c}} 0&0\\ 0&1 \end{array}}
\right),A_4 = \left( {\begin{array}{*{20}{c}} 0&0\\ 0&{ - 1} \end{array}}
\right), \end{aligned}  \label{2d_A2}
\end{equation}%
 we get $\sigma _{2}^{2}\in \Gamma _{2}
\stackrel{\triangle}{=}[-\mathbb{E}{[-X_{2}^{2}]},\mathbb{E}{[X_{2}^{2}]}]$.
On the other hand, by choosing specific matrices  as:
\begin{equation}
\begin{aligned} A_5 = \left( {\begin{array}{*{20}{c}} 0&1\\ 1&0 \end{array}}
\right),A_6 = \left( {\begin{array}{*{20}{c}} 0&-1\\ -1&0 \end{array}}
\right),  \end{aligned}  \label{2d_A12}
\end{equation}%
 we get the uncertainty in the covariance between $X_{1}$, and $X_{2}$,
\begin{equation}
\left\{
\begin{array}{l}
\mathop {{\rm{sup}}}\limits_{b_{12}\in {\Gamma _{12}}%
}b_{12}=\mathbb{E}{[X_{1}X_{2}]}, \\
\mathop {{\rm{inf}}}\limits_{b_{12}\in {\Gamma _{12}}%
}b_{12}=-\mathbb{E}{[-X_{1}X_{2}]}.%
\end{array}%
\right.  \label{2d_uncertain_correlation}
\end{equation}%
Thus, we obtain $b_{12} \in \Gamma_{12}
\stackrel{\triangle}{=} [- \mathbb{E} {[- X_1X_2]}, \mathbb{E} {[X_1X_2]}]$. In summary,   \eqref{G_heat} can be written as follows:
\begin{equation}
\left\{
\begin{array}{l}
u_{t}-\sup\limits_{\substack{ \sigma _{1}^{2}\in \Gamma _{1},\sigma
_{2}^{2}\in \Gamma _{2},  \\ b_{12}\in \Gamma _{12}}}%
\left( \frac{\sigma _{1}^{2}}{2}u_{xx}+\frac{\sigma _{2}^{2}}{2}u_{yy}+b_{12}u_{xy}\right) =0,(t,x,y)\in (0,\infty )\times \mathbb{R%
}\times \mathbb{R}, \\
u(0,x,y)=\varphi (x,y).%
\end{array}%
\right.  \label{2d_G3}
\end{equation}

Specifically, if $X_{1}$ is independent of $X_{2}$ (or if $X_{2}$ is
independent of $X_{1}$ ), we have $\mathbb{E}\left[ X_{1}X_{2}\right] =0$,
and $b_{12}=0$. The G-heat equation \eqref{G_heat} is
referred to as the independent model (IM),
\begin{equation}
\left\{
\begin{array}{l}
u_{t}-\underset{\sigma _{1}^{2}\in \Gamma _{1},\sigma _{2}^{2}\in \Gamma _{2}%
}{\sup }\left( \frac{\sigma _{1}^{2}}{2}u_{xx}+\frac{\sigma _{2}^{2}}{2}%
u_{yy}\right) =0,\left( t,x,y\right) \in (0,\infty )\times \mathbb{R}\times
\mathbb{R}, \\
u(0,x,y)=\varphi \left( x,y\right).%
\end{array}%
\right.  \label{IM}
\end{equation}%

\begin{remark}
 If there is no uncertainty in the volatility $\sigma
_{1}^{2}$ and $\sigma _{2}^{2}$, that is, $-\mathbb{E}[-X_{1}^{2}]=\mathbb{E}[X_{1}^{2}]={\sigma^2_{1}}$, $-\mathbb{E}[-X_{2}^{2}]=\mathbb{E}[X_{2}^{2}]={\sigma^2_{2}}$,
then we can define the uncertainty of `correlation coefficient' as  
$$
\rho  \in  \left[\frac{-\mathbb{E}[-X_{1}X_{2}]}{\sigma_1\sigma_2},\frac{\mathbb{E}[X_{1}X_{2}]}{\sigma_1\sigma_2}\right].
$$
Otherwise, we do not have the concept of correlation coefficient. 
\end{remark}

  \eqref{2d_G3} and \eqref{IM} are   HJB equations of some stochastic control problems,
and numerical methods have been extensively investigated when the sign of $b_{12}$ is definite
\textsuperscript{\cite{clift2008numerical},}\textsuperscript{%
\cite{Forsyth2007},}\textsuperscript{\cite{MF},}\textsuperscript{%
\cite{KSENDAL}}. If $\mathbb{E} {[- X_1 X_2]}* \mathbb{E} {[X_1 X_2]}]>0$, then the sign of $b_{12}$  can be changed. There does not exist any reliable numerical scheme for this case.

\section{A finite difference scheme}

\label{sec:3} \label{section3}   In this
section, we consider the case where the sign of covariance $b_{12}$ in the G-heat
equation \eqref{2d_G3} is uncertain.
\begin{equation}
\left\{
\begin{array}{l}
u_{t}-\sup\limits_{_{\substack{ \sigma _{1}^{2}\in \Gamma _{1},\sigma
_{2}^{2}\in \Gamma _{2},  \\ b_{12}\in \Gamma _{12}}}%
}\left( \frac{\sigma _{1}^{2}}{2}u_{xx}+\frac{\sigma _{2}^{2}}{2}u_{yy}+b_{12}u_{xy}\right) =0,(t,x,y)\in (0,\infty )\times \mathbb{R%
}\times \mathbb{R}, \\
u(0,x,y)=\varphi (x,y),%
\end{array}%
\right.  \label{UVCM}
\end{equation}%
where $\Gamma _{1}=[\underline{\sigma _{1}}^{2},\overline{\sigma _{1}}^{2}],$
$\Gamma _{2}=[\underline{\sigma _{2}}^{2},\overline{\sigma _{2}}^{2}],$  $%
\Gamma _{12}=[\underline{b_{12}},\overline{b_{12}}],$ and specially, 
\begin{equation}
  \underline{b_{12}}< 0 < 
\overline{b_{12}}. \label{b-range}  
\end{equation} Here we used the notations $\overline{\sigma_{i}}^{2}=\mathbb{E}\left[ X_{i}^{2}\right],\underline{\sigma _{i}}^{2}=-\mathbb{E}\left[- X_{i}^{2}\right]$, 
and 
$\overline{b_{12}} =\mathbb{E}{[X_{1}X_{2}]},  
\underline{b_{12}} =-\mathbb{E}{[-X_{1}X_{2}]}$.

It is easy to see that the optimal variance and covariance 
 are reached at the end of the intervals,  depending only on the sign of the second-order partial derivatives, that is,
\begin{equation}
\sigma _{1}^{2}\left( u_{xx}\right) =\left\{
\begin{array}{l}
\overline{\sigma _{1}}^{2}\text{\ \ }u_{xx}\geq 0, \\
\underline{\sigma _{1}}^{2}\text{ \ }u_{xx}<0,%
\end{array}%
\right.
\end{equation}%

\begin{equation}
\sigma _{2}^{2}\left( u_{yy}\right) =\left\{
\begin{array}{l}
\overline{\sigma _{2}}^{2}\text{\ \ }u_{yy}\geq 0, \\
\underline{\sigma _{2}}^{2}\text{ \ }u_{yy}<0,%
\end{array}%
\right.
\end{equation}%
and
\begin{equation}
b_{12}\left( u_{xy}\right) =\left\{
\begin{array}{l}
\overline{b_{12}}\text{ \ \ }u_{xy}\geq 0, \\
\underline{b_{12}}\text{ \ \ }u_{xy}<0.%
\end{array}%
\right. \label{b12}
\end{equation}

For computational purpose, we confine the problem \eqref{UVCM} within a truncated bounded
domain,
\begin{equation*}
0\leq t\leq T\text{ and }\left( x,y\right) \in  {\Omega },\text{ }%
 {\Omega }=\left\{ \left( x,y\right) | |x|<L,\text{ }%
|y|< L \right\},
\end{equation*}
with Dirichlet boundary condition. Subsequently, the problem is reformulated
as
\begin{equation}
\left\{
\begin{array}{l}
u_{t}-\frac{\sigma _{1}^{2}\left( u_{xx}\right) }{2}u_{xx}-\frac{\sigma _{2}^{2}\left( u_{yy}\right) }{2}u_{yy}-b_{12}\left( u_{xy}\right) u_{xy}=0,\text{ }\left( t,x,y\right) \in \left( 0,T\right) \times \Omega , \\
u|_{t=0}=\varphi (x,y), \\
u|_{\left( x,y\right) \in \partial \Omega }=\phi \left(
t,x,y\right).\end{array}\right.  \label{general model}
\end{equation}

\begin{remark}
The Dirichlet boundary condition is imposed on the boundary.  We can expect the errors
incurred by imposing approximate boundary conditions to be small in areas of interest if the
truncated domain is sufficiently large\textsuperscript{\cite{1995CONVERGENCE}}.
\end{remark}

Taking an equi-distance partition with a spatial step size $h=2L/M, \Delta t=T/N$, we have grid points $x_i=-L+i*h, y_j=-L+j*h, t^n=n \Delta t, $ for $i,j=0,\cdots,M,$ and $n=0,\cdots, N$. 

Let $U_{i,j}^{n}$ denote the approximate solution at $\left( t^{n},x_{i},y_{j}\right).$
An implicit scheme for equation \eqref{general model}
reads as, for $n=0, \cdots,N-1$,
\begin{equation}
\left\{
\begin{array}{l}
\delta _{t}U_{i,j}^{n+1}-\frac{\sigma _{1}^{2}\left( \delta
_{x}^{2}U_{i,j}^{n+1}\right) }{2}\delta _{x}^{2}U_{i,j}^{n+1}-\frac{\sigma
_{2}^{2}\left( \delta _{y}^{2}U_{i,j}^{n+1}\right) }{2}\delta
_{y}^{2}U_{i,j}^{n+1}-(b_{12}\delta
_{xy}U)_{i,j}^{n+1}=0,0<i,j<M, \\
U_{i,j}^{0}=\varphi (x_{i},y_{j}),\text{ }i,j=0,...,M, \\
U_{i,j}^{n+1}|_{\left( x_{i},y_{j}\right) \in \partial \Omega }=\phi \left(
t^{n+1},x_{i},y_{j}\right) ,%
\end{array}%
\right.  \label{implicit  correlation2}
\end{equation}%
where
\begin{eqnarray}
\delta _{t}U_{i,j}^{n+1} &=&\frac{U_{i,j}^{n+1}-U_{i,j}^{n}}{\Delta t},\text{
\ }\notag \\
\delta _{x}^{2}U_{i,j}^{n+1} &=&\frac{%
U_{i+1,j}^{n+1}-2U_{i,j}^{n+1}+U_{i-1,j}^{n+1}}{h^{2}},\text{ }\delta
_{y}^{2}U_{i,j}^{n+1}=\frac{U_{i,j+1}^{n+1}-2U_{i,j}^{n+1}+U_{i,j-1}^{n+1}}{%
h^{2}}, \notag\\
\text{ }(b_{12}\delta _{xy}U)_{i,j}^{n+1} &=&\max \left(
\overline{b_{12}}\delta _{xy}^{+}U_{i,j}^{n+1},%
\underline{b_{12}}\delta _{xy}^{-}U_{i,j}^{n+1}\right) ,
\label{cross-approx}\\
\delta _{xy}^{+}U_{i,j}^{n+1} &=&\frac{%
U_{i+1,j+1}^{n+1}+2U_{i,j}^{n+1}+U_{i-1,j-1}^{n+1}-\left(
U_{i+1,j}^{n+1}+U_{i-1,j}^{n+1}+U_{i,j+1}^{n+1}+U_{i,j-1}^{n+1}\right) }{%
2h^{2}},\label{xy+} \\
\delta _{xy}^{-}U_{i,j}^{n+1} &=&\frac{%
U_{i+1,j}^{n+1}+U_{i-1,j}^{n+1}+U_{i,j+1}^{n+1}+U_{i,j-1}^{n+1}-\left(
U_{i+1,j-1}^{n+1}+2U_{i,j}^{n+1}+U_{i-1,j+1}^{n+1}\right) }{2h^{2}}, \label{xy-}
\end{eqnarray}%
and
\begin{equation}
\sigma _{1}^{2}\left( s\right) =\left\{
\begin{array}{l}
\overline{\sigma _{1}}^{2}\text{\ \ \ if }s\geq 0, \\
\underline{\sigma _{1}}^{2}\text{ \ \ if }s<0,%
\end{array}%
\right. \text{, }\sigma _{2}^{2}\left( s\right) =\left\{
\begin{array}{l}
\overline{\sigma _{2}}^{2}\text{\ \ \ if }s\geq 0, \\
\underline{\sigma _{2}}^{2}\text{ \ \ if }s<0.%
\end{array}%
\right.  \label{optimal_sigma}
\end{equation}%

 \begin{remark}
     The approximation for the second-order cross-derivative in \eqref{cross-approx} plays a key role in obtaining a monotonic scheme. As is known, whether $\delta _{xy}^{+}U_{i,j}^{n+1}$ in \eqref{xy+} or $\delta _{xy}^{-}U_{i,j}^{n+1}$ in \eqref{xy-} is applied to approximate the second-order cross-derivative depends on the sign of $b_{12}$, but the sign of $b_{12}$ depends on the sign of $\delta _{xy}U_{i,j}^{n+1}$ from \eqref{b12} and \eqref{b-range}. The choice in \eqref{cross-approx} breaks this cycle of dilemmas. In the cases of $\delta _{xy}^{+}U_{i,j}^{n+1}<0$ and $\delta _{xy}^{-}U_{i,j}^{n+1}>0$, the choice in \eqref{cross-approx} yields $b_{12}*\delta _{xy}U_{i,j}^{n+1}<0$, which breaks the constraint $b_{12}*u_{xy} \geq 0$ from \eqref{b12}. However, with this choice, the scheme \eqref{implicit  correlation2} is monotonic and works well. In fact, in the cases of $\delta _{xy}^{+}U_{i,j}^{n+1}<0$ and $\delta _{xy}^{-}U_{i,j}^{n+1}>0$, we have $u_{xy}\approx 0$, which means that the second-order cross-derivative is ignorable.
 \end{remark}

Since \eqref{implicit  correlation2} is a nonlinear system, an inner iteration is needed to obtain the solution $U^n_{i,j}$  in each time step. Let $U_{i,j}^{n+1,k}$ be the $k^{th}$
estimate for $U_{i,j}^{n+1}$,   $U_{i,j}^{n+1,k+1}$  is given by the following Picard's iteration,
\begin{equation}
\left\{
\begin{array}{l}
\delta _{t}U_{i,j}^{n+1,k+1}-\frac{\sigma _{1}^{2}\left( \delta
_{x}^{2}U_{i,j}^{n+1,k}\right) }{2}\delta _{x}^{2}U_{i,j}^{n+1,k+1}-\frac{%
\sigma _{2}^{2}\left( \delta _{y}^{2}U_{i,j}^{n+1,k}\right) }{2}\delta
_{y}^{2}U_{i,j}^{n+1,k+1}  -\left( b_{12}\right)_{i,j} ^{k}\delta _{xy}^{\alpha_k}U_{i,j}^{n+1,k+1}=0, \\
U_{i,j}^{0}=\varphi (x_{i},y_{j}), \\
U_{i,j}^{n+1,k+1}|_{\left( x_{i},y_{j}\right) \in \partial \Omega }=\phi
\left( t^{n+1},x_{i},y_{j}\right) ,%
\end{array}%
\right.  \label{wxpand22}
\end{equation}%
with $U_{i,j}^{n+1,0}=U_{i,j}^{n},$ where
\begin{equation}
\left( b_{12}\right)_{i,j} ^{k}=\left\{
\begin{array}{l}
\overline{b_{12}}, \text{ \ \ if\ }\overline{b_{12}}\delta _{xy}^{+}U_{i,j}^{n+1,k}\geq \underline{b_{12}}\delta _{xy}^{-}U_{i,j}^{n+1,k}, \\
\underline{b_{12}}, \text{ \ \ if\ }\overline{b_{12}}\delta _{xy}^{+}U_{i,j}^{n+1,k} < \underline{b_{12}}\delta _{xy}^{-}U_{i,j}^{n+1,k},%
\end{array}%
\right.  \label{wxpand32}
\end{equation}%
and
\begin{equation}
\alpha_k =\left\{
\begin{array}{l}
+\text{ \ \ \ if }\left( b_{12}\right) ^{k}=\overline{%
b_{12}}, \\
-\text{ \ \ \ if }\left( b_{12}\right) ^{k}=\underline{%
b_{12}}.%
\end{array}%
\right.  \label{wxpand222}
\end{equation}

In the next section, we discuss the convergence of the iterative scheme \eqref{wxpand22} and some theoretical convergence issues for
the discrete scheme \eqref{implicit  correlation2}. 

\section{Numerical analysis}

\label{sec:4}
 The implicit scheme \eqref{wxpand22} leads to a nonlinear
algebraic system which must be solved by an
inner iteration at each time step. In this section, we first prove the
convergence of the inner iteration and then check the properties of consistence, stability, and monotonicity.

\subsection{ Convergence of inner iteration}

 
We first present an assumption,  then prove the convergence of the inner iteration. 
 
\begin{assumption}
\label{assume} The covariance matrix $\left(
\begin{array}{ll}
\text{ }\sigma _{1}^{2} & b_{12} \\
b_{12} & \text{ }\sigma _{2}^{2}\end{array}\right) $ is a diagonally dominated, where $\sigma _{1}^{2}\in \left\{
\underline{\sigma _{1}}^{2},\overline{\sigma _{1}}^{2}\right\}$, $\sigma
_{2}^{2}\in \left\{ \underline{\sigma _{2}}^{2},\overline{\sigma _{2}}^{2} \right\}$ and $b_{12}\in \left\{  \underline{b_{12}},\overline{b_{12}}\right\}$.
\end{assumption}

\begin{proposition}
\textrm{(Maximum principle)}\label{EP} If covariance matrix $\left(
\begin{array}{ll}
\text{ }\sigma _{1}^{2} & b_{12} \\
b_{12} & \text{ }\sigma _{2}^{2}\end{array}\right) $ is diagonally dominated, and $\{U_{i,j}^{n}\}$ satisfies  
\begin{equation}
\mathbbm{L}_{h}U_{i,j}^{n}\equiv \delta _{t}U_{i,j}^{n}-\frac{\sigma _{1}^{2}}{2%
}\delta _{x}^{2}U_{i,j}^{n}-\frac{\sigma _{2}^{2}}{2}\delta
_{y}^{2}U_{i,j}^{n}-b_{12}\delta _{xy}^{\alpha
}U_{i,j}^{n}\geq 0 (\leq 0), \ n=1,...,N,\ 0<i,j<M,    \label{linear eq}
\end{equation}%
where
\begin{eqnarray} \label{rule-of-alpha}
\alpha &=&\left\{
\begin{array}{l}
+\text{ \ \ if }b_{12}\geq 0, \\
-\text{ \ \ if }b_{12}<0,%
\end{array}%
\right.
\end{eqnarray}%
 then the  minimum  (maximum) of $\{U_{i,j}^{n}\}$ can only be achieved at the initial or boundary points, unless $\{U_{i,j}^{n}\}$ is constant.
\end{proposition}

Re-formulate the right-hand side of system \eqref{linear eq} into an operator form as
$$\mathbbm{L}_{h}U^{n} = \left(\frac{1}{\Delta t} I -\frac{\sigma _{1}^{2}}{2%
}\delta _{x}^{2}-\frac{\sigma _{2}^{2}}{2}\delta
_{y}^{2}-b_{12}\delta _{xy}^{\alpha }\right )U^{n} - \frac{1}{\Delta t} U^{n-1}  \equiv A U^n - \frac{1}{\Delta t} U^{n-1}. $$
It is easy to check that, thanks to the appropriate choice of the approximation to the second-order cross-derivative,  $A=(a_{ij})$ is an M-matrix, that is, $a_{ii}>0, a_{ij}
\leq 0$, for $i\not=j$, and $\sum_{j\not =i} |a_{ij}|< a_{ii}$. Then the above proposition follows.

\begin{theorem}
\label{CII2} \textrm{(Convergence of  inner iteration)} If Assumption \ref%
{assume} holds, then for any initial guess $U^{n+1,0}$, the
iterative sequence $\{ U^{n+1,k}\}_{k>0}$in \eqref {wxpand22} is bounded and monotonically increasing, so  converges to the unique
solution to \eqref{implicit  correlation2}.  
\end{theorem}

\begin{proof}
Without loss of generality (WLOG), we   assume that there is no uncertainty in ${\sigma }_{1}^{2}$ and $%
{\sigma }_{2}^{2},$ and pay more attentions on the second-order cross-derivative. 
Denote by $\overline{U}^k=U^{n+1,k}$ and let $%
W^{k}=\overline{U}^{k+1}-\overline{U}^{k},$ $k\geq 1,$ then $W_{i,j}^{k}$
satisfies the following difference equation
\begin{equation}
\left\{
\begin{array}{l} 
\frac{W_{i,j}^{k}}{\Delta t} -  \frac{\sigma_{1}^{2}}{2} \delta _{x}^{2}W_{i,j}^{k}-\frac{\sigma_{2}^{2}}{2}   \delta
_{y}^{2}W_{i,j}^{k}-\left(( b_{12})_{i,j}^{k}\delta
_{xy}^{\alpha _{k}}\overline{U}^{k+1}_{i,j}- (b_{12})_{i,j}^{k-1}\delta
_{xy}^{\alpha _{k-1}}\overline{U}^{k}_{i,j}\right)=0, \\
W_{i,j}^{k}|_{(x_{i},y_{j})\in \partial \Omega }=0.%
\end{array}%
\right.  \label{mono1}
\end{equation}%

 If $( b_{12})_{i,j}^{k}=( b_{12})_{i,j}^{k-1}$, for example, equals to $\overline{b_{12}}$, we have 
 \begin{equation} \label{iter1}
\frac{W_{i,j}^{k}}{\Delta t} -  \frac{\sigma_{1}^{2}}{2} \delta _{x}^{2}W_{i,j}^{k}-\frac{\sigma_{2}^{2}}{2}   \delta
_{y}^{2}W_{i,j}^{k}-  \overline{b_{12}}\delta
_{xy}^{+}W^{k}_{i,j}=0.\end{equation}
  If $( b_{12})_{i,j}^{k}\not=( b_{12})_{i,j}^{k-1}$, for example, $$
\left( b_{12}\right)_{i,j} ^{k} =\underline{b_{12}},  \ \ 
\left( b_{12}\right)_{i,j} ^{k-1} =\overline{b_{12}},
$$
which means that, from \eqref{wxpand32}, $\underline{b_{12}}\delta _{xy}^{-}\overline{%
U}_{ij}^{k} \geq \overline{b_{12}}\delta _{xy}^{+}\overline{U}_{ij}^{k}$, 
 then we have  \begin{equation} \label{iter2}
\frac{W_{i,j}^{k}}{\Delta t} -  \frac{\sigma_{1}^{2}}{2} \delta _{x}^{2}W_{i,j}^{k}-\frac{\sigma_{2}^{2}}{2}   \delta
_{y}^{2}W_{i,j}^{k}-  \underline{b_{12}}\delta
_{xy}^{-}W^{k}_{i,j}=\underline{b_{12}}\delta _{xy}^{-}\overline{%
U}_{ij}^{k}-\overline{b_{12}}\delta _{xy}^{+}\overline{U%
}_{ij}^{k} \geq 0.\end{equation}
Otherwise, if  $$
\left( b_{12}\right)_{i,j} ^{k} =  \overline{b_{12}},  \ \ 
\left( b_{12}\right)_{i,j} ^{k-1} =\underline{b_{12}},
$$
which means that, from \eqref{wxpand32},   $ \overline{b_{12}}\delta _{xy}^{+}\overline{U}_{ij}^{k} \geq \underline{b_{12}}\delta _{xy}^{-}\overline{%
U}_{ij}^{k}$, 
 then we have \begin{equation} \label{iter3}
\frac{W_{i,j}^{k}}{\Delta t} -  \frac{\sigma_{1}^{2}}{2} \delta _{x}^{2}W_{i,j}^{k}-\frac{\sigma_{2}^{2}}{2}   \delta
_{y}^{2}W_{i,j}^{k}-  \overline{b_{12}}\delta
_{xy}^{+}W^{k}_{i,j}=\overline{b_{12}}\delta _{xy}^{+}\overline{%
U}_{ij}^{k}-\underline{b_{12}}\delta _{xy}^{-}\overline{U%
}_{ij}^{k} \geq 0.\end{equation}
Combining \eqref{iter1}-\eqref{iter3}, we have, for any $ 0<i,j<M,$
$$\frac{W_{i,j}^{k}}{\Delta t} -  \frac{\sigma_{1}^{2}}{2} \delta _{x}^{2}W_{i,j}^{k}-\frac{\sigma_{2}^{2}}{2}   \delta
_{y}^{2}W_{i,j}^{k}-  (b_{12})^k_{i,j}\delta
_{xy}^{\alpha_k}W^{k}_{i,j} \geq 0, $$
subject to the boundary condition $W^k_{i,j}=0$, for $(x_i,y_j)\in \partial \Omega$.
From the maximum principle, we have $W^k_{i,j}\geq 0$, for $ 0<i,j<M$, that is, $\overline{U}^{k+1}_{i,j}\geq \overline{U}^{k}_{i,j}, k\geq 1$.

So $\{\overline{U}^{k}\}_{k>0}$ is a monotonic increasing sequence. 
Now we check the boundedness of the sequence. From Proposition \ref{EP}, the maximum principle is valid for the system \eqref{wxpand22}. It follows that
\begin{equation}
\left\vert \left\vert \overline{U}^{k+1}\right\vert \right\vert _{\infty
}\leq \max\left( \left\vert \left\vert \varphi \right\vert \right\vert
_{\infty }, \max_n \left\vert \left\vert \phi^n \right\vert \right\vert _{\infty, \partial \Omega}\right).
\end{equation}%
Consequently, as a monotonic and bounded sequence,  $ \overline{U}^{k} $ converges.
\end{proof}

\subsection{Monotonicity and Convergence of implicit scheme \texorpdfstring{\eqref{implicit correlation2}}{implicit scheme}}

From the work of Barles and Souganidis\textsuperscript{	%
\cite{G1991Convergence}}, we know that numerical solution of \eqref{implicit correlation2} converges to the viscosity solution of the fully nonlinear
PDE \eqref{general model} if the scheme \eqref{implicit correlation2} is consistent, stable (in the $%
l_{\infty }$ norm) and monotone.  




\begin{lemma}
\label{LConsistency2}\textrm{(Consistency) }The implicit
scheme \eqref{implicit  correlation2} is consistent.
\end{lemma}
\begin{proof}
  It is easy to check that the difference equation in \eqref{implicit  correlation2} tends to the G-equation \eqref{UVCM} as $h,\Delta t\rightarrow 0$, since the `$\sup$' operation is continuous.  
\end{proof}

 We now give the definition of monotonicity. Denote by $$g_{i,j} = g\left(U^{n+1}_{i,j}, U^n_{i,j}, \{ U^{n+1}_{k,l} \}_{(k,l)\in N_{i,j}}\right)$$ the left-hand side of the difference equation \eqref{implicit  correlation2}. Here, $N_{i,j}=\{(k,l)\not = (i,j) :\ |k-i|\leq 1, |l-j|\leq 1\}$ presents the set of all nearest-neighbor indexes of $(i,j)$.
 \begin{definition} \textrm{(Monotonicity)} \label{def-mono}
    The scheme \eqref{implicit  correlation2} is monotone if for all  $(i,j)$, \begin{equation}\label{g-diag}
        g_{i,j}\left(U^{n+1}_{i,j}+\epsilon^{n+1}_{i,j}, U^n_{i,j}, \{U^{n+1}_{k,l}\}_{(k,l)\in N_{i,j}}\right)\geq  g_{i,j}\left(U^{n+1}_{i,j}, U^n_{i,j}, \{U^{n+1}_{k,l}\}_{(k,l)\in N_{i,j}}\right), \ \ \forall \epsilon^{n+1}_{i,j}\geq 0, 
    \end{equation}
    and 
    \begin{equation} \label{g-off-diag}
      g_{i,j}\left(U^{n+1}_{i,j}, U^n_{i,j}+\epsilon^{n}_{i,j}, \{U^{n+1}_{k,l}+\epsilon^{n+1}_{k,l}\}_{(k,l)\in N_{i,j}}\right)\leq  g_{i,j}\left(U^{n+1}_{i,j}, U^n_{i,j}, \{U^{n+1}_{k,l}\}_{(k,l)\in N_{i,j}}\right), \  \forall \epsilon^{n}_{i,j}, \epsilon^{n+1}_{k,l}\geq 0. 
    \end{equation}
 \end{definition}
 
\begin{lemma}\label{lem-mono}
\textrm{(Monotonicity)} If Assumption \ref{assume} holds, then the
implicit scheme \eqref{implicit  correlation2} is monotone.
\end{lemma}
\begin{proof}
    We first consider the perturbation on $U^{n+1}_{i,j}$, and denote it by $\widetilde{U}^{n+1}_{i,j}=U^{n+1}_{i,j}+\epsilon^{n+1}_{i,j}$, for $\epsilon^{n+1}_{i,j}\geq 0$. We also denote by $\widetilde{U}^{n+1}_{k,l}=U^{n+1}_{k,l}$, for $(k,l)\in N_{i,j}$.
    Then the difference between the two sides of the inequality \eqref{g-diag} is \begin{eqnarray}
     T&:=&    g_{i,j}\left(\widetilde{U}^{n+1}_{i,j}, U^n_{i,j}, \{\widetilde{U}^{n+1}_{k,l}\}_{(k,l)\in N_{i,j}}\right)- g_{i,j}\left(U^{n+1}_{i,j}, U^n_{i,j}, \{U^{n+1}_{k,l}\}_{(k,l)\in N_{i,j}}\right)\notag\\
     &=& \frac{\epsilon^{n+1}_{i,j}}{\Delta t} -\left( \frac{\widetilde{\sigma} _{1}^{2}  }{2}\delta _{x}^{2}\widetilde{U}_{i,j}^{n+1} - \frac{\widehat{\sigma} _{1}^{2}  }{2}\delta _{x}^{2}U_{i,j}^{n+1}\right)
     -\left( \frac{\widetilde{\sigma} _{2}^{2}  }{2}\delta _{y}^{2}\widetilde{U}_{i,j}^{n+1} - \frac{\widehat{\sigma} _{2}^{2}  }{2}\delta _{y}^{2}U_{i,j}^{n+1}\right)\notag\\
&&\ \   \ \  \ \ \  -\left( \widetilde{b_{12}}  \delta^{\tilde{\alpha}}
_{xy}\widetilde{U}_{i,j}^{n+1}- \widehat{b}_{12}\delta^\alpha
_{xy}U_{i,j}^{n+1}\right)\notag\\
&:=&  \frac{\epsilon^{n+1}_{i,j}}{\Delta t}+T_1+T_2+T_3,\label{T-diag}
    \end{eqnarray}
    where $\widetilde{\sigma}_1=\sigma_1(\delta _{x}^{2}\widetilde{U}_{i,j}^{n+1}), \widehat{\sigma}_1=\sigma_1(\delta _{x}^{2}{U}_{i,j}^{n+1})$, $\widetilde{\sigma}_2=\sigma_2(\delta _{y}^{2}\widetilde{U}_{i,j}^{n+1}), \widehat{\sigma}_2=\sigma_1(\delta _{y}^{2}{U}_{i,j}^{n+1})$ are defined by \eqref{optimal_sigma}, $ \widetilde{b_{12}}, \widehat{b}_{12}$ and  the index $\tilde{\alpha}$ and $\alpha$ are determined by the rule \eqref{cross-approx} and \eqref{rule-of-alpha}. 
    
    The term $T_1$ can be treated as \begin{eqnarray}
        T_1&=& -\frac{\widetilde{\sigma} _{1}^{2}  }{2}\left(\delta _{x}^{2}\widetilde{U}_{i,j}^{n+1} -\delta _{x}^{2}U_{i,j}^{n+1} \right)+ \left( \frac{\widehat{\sigma} _{1}^{2}  }{2}\delta _{x}^{2}U_{i,j}^{n+1} -\frac{\widetilde{\sigma} _{1}^{2}  }{2}\delta _{x}^{2} {U}_{i,j}^{n+1}\right)\notag\\
        &\geq & \frac{\widetilde{\sigma} _{1}^{2}}{h^2}\epsilon
_{i,j}^{n+1},
    \end{eqnarray}
    since $\frac{\widehat{\sigma} _{1}^{2} }{2} \delta _{x}^{2}U_{i,j}^{n+1} = \sup_{\sigma_1} \frac{\sigma _{1}^{2}  }{2} \delta _{x}^{2}U_{i,j}^{n+1}\geq \frac{\widetilde{\sigma} _{1}^{2}  }{2}\delta _{x}^{2} {U}_{i,j}^{n+1}$.
    Similarly, we have \begin{equation}
        T_2 \geq \frac{\widetilde{\sigma} _{2}^{2}}{h^2}\epsilon_{i,j}^{n+1}.
    \end{equation}
    We now turn to the term $T_3$. 
    \begin{eqnarray}
        T_3&=&  - \widetilde{b_{12}} ( \delta^{\tilde{\alpha}}
_{xy}\widetilde{U}_{i,j}^{n+1} -\delta^{\tilde{\alpha}}
_{xy}{U}_{i,j}^{n+1})+(\widehat{b}_{12}\delta^\alpha
_{xy}U_{i,j}^{n+1}- \widetilde{b_{12}}\delta^{\tilde{\alpha}}
_{xy}{U}_{i,j}^{n+1})\notag\\
&\geq& - \frac{|\widetilde{b_{12}}|}{h^2}\epsilon_{i,j}^{n+1},
    \end{eqnarray}
    where we have used the fact that $\widehat{b}_{12}\delta^\alpha
_{xy}U_{i,j}^{n+1} = \max\{\overline{b_{12}}\delta^+
_{xy}U_{i,j}^{n+1}, \underline{b_{12}}\delta^-
_{xy}U_{i,j}^{n+1} \}$ $\geq \widetilde{b_{12}}\delta^{\tilde{\alpha}}
_{xy}{U}_{i,j}^{n+1}$ and the definitions \eqref{xy+} and \eqref{xy-}.

Substituting the above three terms into \eqref{T-diag}, we get, by Assumption \eqref{assume}, that
$$ T \geq \epsilon_{i,j}^{n+1} \left(\frac{1}{\Delta t} + \frac{\widetilde{\sigma}^2_1+\widetilde{\sigma}^2_2-|\widetilde{b_{12}}|}{h^2}\right)\geq 0. $$ 
So \eqref{g-diag} is valid.

We now turn to the perturbation on $U_{k,l}^{n+1}$  and $U^n_{i,j}$. Denote by $\widetilde{U}^{n+1}_{k,l}=U^{n+1}_{k,l}+\epsilon^{n+1}_{k,l}$, for $\epsilon^{n+1}_{k,l}\geq 0$ and $(k,l)\in N_{i,j}$. We also denote by $\widetilde{U}^{n}_{i,j}=U^{n}_{i,j}+ \epsilon^n_{i,j}$, for $ \epsilon^n_{i,j}\geq 0$ and $\widetilde{U}^{n+1}_{i,j}=U^{n+1}_{i,j}$.
    Then the difference between the two sides of the inequality \eqref{g-off-diag} is \begin{eqnarray}
     T'&:=&    g_{i,j}\left(\widetilde{U}^{n+1}_{i,j}, \widetilde{U}^n_{i,j}, \{\widetilde{U}^{n+1}_{k,l}\}_{(k,l)\in N_{i,j}}\right)- g_{i,j}\left(U^{n+1}_{i,j}, U^n_{i,j}, \{U^{n+1}_{k,l}\}_{(k,l)\in N_{i,j}}\right)\notag\\
    &=& -\frac{\epsilon^{n}_{i,j}}{\Delta t} -\left( \frac{\widetilde{\sigma} _{1}^{2}  }{2}\delta _{x}^{2}\widetilde{U}_{i,j}^{n+1} - \frac{\widehat{\sigma} _{1}^{2}  }{2}\delta _{x}^{2}U_{i,j}^{n+1}\right)
     -\left( \frac{\widetilde{\sigma} _{2}^{2}  }{2}\delta _{y}^{2}\widetilde{U}_{i,j}^{n+1} - \frac{\widehat{\sigma} _{2}^{2}  }{2}\delta _{y}^{2}U_{i,j}^{n+1}\right)\notag\\
&&\ \   \ \  \ \ \  -\left( \widetilde{b_{12}}  \delta^{\tilde{\alpha}}
_{xy}\widetilde{U}_{i,j}^{n+1}- \widehat{b}_{12}\delta^\alpha
_{xy}U_{i,j}^{n+1}\right)\notag\\
&:=&  -\frac{\epsilon^{n}_{i,j}}{\Delta t}+T_4+T_5+T_6,\label{T-off-diag}
    \end{eqnarray}
   where $\widetilde{\sigma}_1=\sigma_1(\delta _{x}^{2}\widetilde{U}_{i,j}^{n+1}), \widehat{\sigma}_1=\sigma_1(\delta _{x}^{2}{U}_{i,j}^{n+1})$, $\widetilde{\sigma}_2=\sigma_2(\delta _{y}^{2}\widetilde{U}_{i,j}^{n+1}), \widehat{\sigma}_2=\sigma_1(\delta _{y}^{2}{U}_{i,j}^{n+1})$ are defined by \eqref{optimal_sigma}, $ \widetilde{b_{12}}, \widehat{b}_{12}$ and  the index $\tilde{\alpha}$ and $\alpha$ are determined by the rule \eqref{cross-approx} and \eqref{rule-of-alpha}.
    
    The term $T_4$ can be treated as \begin{eqnarray}
        T_4&=& (-\frac{\widetilde{\sigma} _{1}^{2}  }{2}\delta _{x}^{2}\widetilde{U}_{i,j}^{n+1} + \frac{\widehat{\sigma} _{1}^{2}  }{2}\delta _{x}^{2} \widetilde{U}_{i,j}^{n+1} )+\frac{\widehat{\sigma} _{1}^{2}  }{2} (\delta _{x}^{2}U_{i,j}^{n+1} - \delta _{x}^{2} \widetilde{U}_{i,j}^{n+1})\notag\\
        &\leq & -\frac{ \widehat{\sigma} _{1}^{2}}{2h^2} \sum_{|k-i|=1}\epsilon
_{k,j}^{n+1},
    \end{eqnarray}
    since $\frac{\widetilde{\sigma} _{1}^{2} }{2} \delta _{x}^{2}\widetilde{U}_{i,j}^{n+1} = \sup_{\sigma_1} \frac{\sigma _{1}^{2}  }{2} \delta _{x}^{2}\widetilde{U}_{i,j}^{n+1}\geq \frac{\widehat{\sigma} _{1}^{2}  }{2}\delta _{x}^{2} \widetilde{U}_{i,j}^{n+1}$.
    Similarly, we have \begin{equation}
        T_5 \leq -\frac{ \widehat{\sigma} _{2}^{2}}{2h^2} \sum_{|l-j|=1}\epsilon_{i,l}^{n+1}.
    \end{equation}
   The term $T_6$ can be treated as     \begin{eqnarray}
        T_6&=&  \left(- \widetilde{b_{12}}  \delta^{\tilde{\alpha}}_{xy}\widetilde{U}_{i,j}^{n+1} + \widehat{b}_{12} \delta^{{\alpha}}
_{xy}\widetilde{U}_{i,j}^{n+1}\right) + 
 \widehat{b}_{12}\left(
\delta^{{\alpha}}
_{xy}{U}_{i,j}^{n+1}- \delta^{{\alpha}}
_{xy}\widetilde{U}_{i,j}^{n+1}\right) 
\leq  \widehat{b}_{12}\left(
\delta^{{\alpha}}
_{xy}{U}_{i,j}^{n+1}- \delta^{{\alpha}}
_{xy}\widetilde{U}_{i,j}^{n+1}\right) \notag\\
&=&  -\frac{|\widehat{b}_{12}|}{2h^2}*\left((\epsilon_{i+1,j+1}^{n+1} +\epsilon_{i-1,j-1}^{n+1})*H(\widehat{b}_{12}) + (\epsilon_{i+1,j-1}^{n+1} +\epsilon_{i-1,j+1}^{n+1})*H(-\widehat{b}_{12})\right)\notag\\
& &  +\frac{|\widehat{b}_{12}|}{2h^2}\left(\sum_{|k-i|=1}\epsilon_{k,j}^{n+1}+ \sum_{|l-j|=1}\epsilon_{i,l}^{n+1}\right),
    \end{eqnarray}
    where  we have used  the fact that $\widetilde{b_{12}}  \delta^{\tilde{\alpha}}_{xy}\widetilde{U}_{i,j}^{n+1} = \max\{\overline{b_{12}}\delta^+
_{xy}\widetilde{U}_{i,j}^{n+1}, \underline{b_{12}}\delta^-
_{xy}\widetilde{U}_{i,j}^{n+1} \}$ $\geq \widehat{b}_{12}\delta^{{\alpha}}
_{xy}\widetilde{U}_{i,j}^{n+1}$ and the definitions \eqref{xy+} and \eqref{xy-}. $H(\cdot )$ is the Heaviside function.
 
Substituting the above three terms into \eqref{T-off-diag}, we get,  by Assumption \eqref{assume}, that
$$T'\leq  -\frac{ \widehat{\sigma} _{1}^{2}-|\widehat{b}_{12}|}{2h^2} \sum_{|k-i|=1}\epsilon_{k,j}^{n+1}-\frac{ \widehat{\sigma} _{2}^{2}-|\widehat{b}_{12}|}{2h^2} \sum_{|l-j|=1}\epsilon_{i,l}^{n+1}\leq 0. $$
So \eqref{g-off-diag} is true, and the scheme \eqref{implicit  correlation2} is monotone.
\end{proof}

\begin{lemma}
\label{S} \textrm{(Stability)} If Assumption \ref{assume} holds, then the
fully implicit scheme \eqref{implicit  correlation2} is stable, in
the sense that 
\begin{equation}
\max_{n}\left\vert \left\vert U^{n}\right\vert \right\vert _{\infty }\leq \max \left( \left\vert \left\vert \varphi \right\vert \right\vert _{\infty
},\max_{n}\left\vert \left\vert \phi ^{n}\right\vert \right\vert _{\infty, \partial\Omega }\right).  \label{uncondition2}
\end{equation}%
\end{lemma}

\begin{proof}
With the diagonal-dominated assumption, the maximum principle is valid for the system \eqref{implicit  correlation2}. The estimate on $l_\infty$ \eqref{uncondition2} follows directly.
\end{proof}

We also have the following comparison principle.
\begin{lemma}
\label{comparison principle} \textrm{(Comparison Principle)} If Assumption \ref{assume} holds, $\{U_{i,j}^{n+1}\}$   satisfies 
\eqref{implicit
		correlation2}, $\{V_{i,j}^{n+1}\}$ satisfies the same difference equation
\begin{equation*}
\delta _{t}V_{i,j}^{n+1}- \frac{\sigma _{1}^{2}
\left(\delta_{x}^{2}V_{i,j}^{n+1} \right)}{2}\delta _{x}^{2}V_{i,j}^{n+1} -
\frac{\sigma _{2}^{2} \left(\delta_{y}^{2}V_{i,j}^{n+1} \right)}{2}\delta
_{y}^{2}V_{i,j}^{n+1} -(b_{12}\delta
_{xy}V)_{i,j}^{n+1}=0,
\end{equation*}
but  subject to  different initial and boundary value such that
\begin{equation*}
\left\{
\begin{array}{l}
U_{i,j}^{0}\geq V_{i,j}^{0}, \\
U_{i,j}^{n+1}|_{\left( x_{i},y_{j}\right) \in \partial \Omega }\geq
V_{i,j}^{n+1}|_{\left( x_{i},y_{j}\right) \in \partial \Omega },%
\end{array}
\right.
\end{equation*}
then  $U_{i,j}^{n+1}\geq V_{i,j}^{n+1}$,  for all $i,j=1,...,M-1$, and $n<N$. 
\end{lemma}
\begin{proof}
  This lemma can be proved by the techniques used in the proof of Theorem \ref{CII2} and Lemma \ref{lem-mono}. Consider the governing system for $W^n_{i,j}:=U^n_{i,j}-V^n_{i,j}$ and apply the maximum principle.
\end{proof}

\begin{theorem}
\textrm{(Convergence to the viscosity solution) }Let Assumption \ref{assume}
hold, then the  implicit scheme \eqref{implicit  correlation2}
converges to the viscosity solution of the non-linear PDE
\eqref{general
model}.
\end{theorem}

\begin{proof}
Since the  scheme \eqref{implicit  correlation2} is consistent, $%
l_{\infty }-$stable, and monotone, the convergence follows from the results of Barles and Souganidis \cite%
{G1991Convergence} directly.
\end{proof}

\begin{remark}
If Assumption \ref{assume} does not hold, we can also use the idea of a wide
stencil (see Ma and Forsyth\textsuperscript{\cite{MF}}) to construct a
numerical scheme with relaxed conditions, and we leave it to future research.
\end{remark}

\begin{theorem} \label{rate}
\textrm{(Rate of convergence)} Let $u$ be the viscosity solution of equation \eqref{general model}, $U$ be the numerical solution of equation \eqref{implicit correlation2}.  
If Assumption \ref{assume} holds and there exists some $\beta \in \left( 0,1\right)$, such that $u\in C^{1+\beta /2,2+\beta }\left( \left[0,T\right] \times \Omega\right)$,    then 
\begin{equation}
\max_n\left\vert \left\vert u^{n}-U^{n}\right\vert \right\vert
_{\infty }\leq C\left( \Delta t^{\frac{\beta }{2}}+h^{\beta }\right), \end{equation}
where $C$ is a positive constant independent of $\Delta t$ and $h$.
\end{theorem}
\begin{proof}
 Approximating the derivatives by  corresponding difference quotients in \eqref{general model}, we obtain  
\begin{eqnarray}
0&=& \delta _{t}u_{i,j}^{n}+R^n_{t}-\sup\limits_{_{\sigma _{1}^{2}\in \Gamma
_{1}}}\left( \frac{\sigma _{1}^{2}}{2}\left( \delta
_{x}^{2}u_{i,j}^{n}+R^n_{x}\right) \right) -\sup\limits_{_{\sigma _{2}^{2}\in
\Gamma _{2}}}\left( \frac{\sigma _{2}^{2}}{2}\left( \delta
_{y}^{2}u_{i,j}^{n}+R^n_{y}\right) \right)\notag\\
& & -\max\{\overline{b_{12}}\left( \delta^+ _{xy}u_{i,j}^n+R^+_{xy}\right), \underline{b_{12}}\left( \delta^- _{xy}u_{i,j}^n+R^-_{xy}\right) \},
\end{eqnarray}
 since $\sup\limits_{b_{12}\in \Gamma_{12}} (b_{12} u_{xy})= \max\{\overline{b_{12}} u_{xy}, \ \underline{b_{12}} u_{xy}\}$. Here we have the truncation error terms $$R^n_t=u^n_{t}-\delta_{t}u^n=O((\Delta t)^{\beta/2}), R^n_x=u^n_{xx}-\delta^2_x u^n =O(h^\beta), R^n_y=u^n_{yy}-\delta^2_y u^n=O(h^\beta), $$
 and $$R^+_{xy}=u_{xy}-\delta^+ _{xy}u = O(h^\beta), R^-_{xy}=u_{xy}-\delta^- _{xy}u = O(h^\beta).$$
Thanks to $\sup(f+g)\leq \sup f+ \sup g$, we have
\begin{equation} \label{u-up}
 \delta _{t}u_{i,j}^{n}-\sup\limits_{_{\sigma _{1}^{2}\in \Gamma
_{1}}}\left( \frac{\sigma _{1}^{2}}{2}\delta _{x}^{2}u_{i,j}^{n}\right)
-\sup\limits_{_{\sigma _{2}^{2}\in \Gamma _{2}}}\left( \frac{\sigma _{2}^{2}%
}{2}\delta _{y}^{2}u_{i,j}^{n}\right) -\max\{\overline{b_{12}}  \delta^+ _{xy}u_{i,j}^n, \underline{b_{12}}  \delta^- _{xy}u_{i,j}^n  \}\leq R^n_{up},
\end{equation}
where \begin{equation}
R^n_{up}= -R^n_t+\sup\limits_{_{\sigma _{1}^{2}\in \Gamma
_{1}}}\left( \frac{\sigma _{1}^{2}}{2} R^n_x\right)
+\sup\limits_{_{\sigma _{2}^{2}\in \Gamma _{2}}}\left( \frac{\sigma _{2}^{2}%
}{2}R^n_y\right) +\max\{\overline{b_{12}} R^+_{xy}, \underline{b_{12}}R^-_{xy} \}.
\end{equation}
 By the fact $\sup (f+g)\geq \sup f + \inf g$, we have 
\begin{equation} \label{u-low}
 \delta _{t}u_{i,j}^{n}-\sup\limits_{_{\sigma _{1}^{2}\in \Gamma
_{1}}}\left( \frac{\sigma _{1}^{2}}{2}\delta _{x}^{2}u_{i,j}^{n}\right)
-\sup\limits_{_{\sigma _{2}^{2}\in \Gamma _{2}}}\left( \frac{\sigma _{2}^{2}%
}{2}\delta _{y}^{2}u_{i,j}^{n}\right) -\max\{\overline{b_{12}}  \delta^+ _{xy}u_{i,j}^n, \underline{b_{12}}  \delta^- _{xy}u_{i,j}^n  \}\geq R^n_{low},
\end{equation}
where \begin{equation}
R^n_{low}= -R^n_t+\inf\limits_{_{\sigma _{1}^{2}\in \Gamma
_{1}}}\left( \frac{\sigma _{1}^{2}}{2} R^n_x\right)
+\inf\limits_{_{\sigma _{2}^{2}\in \Gamma _{2}}}\left( \frac{\sigma _{2}^{2}%
}{2}R^n_y\right) +\min\{\overline{b_{12}} R^+_{xy}, \underline{b_{12}}R^-_{xy} \}.
\end{equation}

Set $V_{i,j}^{n}=u_{i,j}^{n}-U_{i,j}^{n}-  t^{n}*\max\limits_{n} \|R^n_{up}\|_\infty$, then we have, by the fact $\sup (f-g)\geq \sup f -\sup g$, that, for $0<i,j<M,$
\begin{eqnarray}
&&\delta _{t}V_{i,j}^{n}-\sup\limits_{_{\sigma _{1}^{2}\in \Gamma _{1}}}\left(
\frac{\sigma _{1}^{2}}{2}\delta _{x}^{2}V_{i,j}^{n}\right)
-\sup\limits_{_{\sigma _{2}^{2}\in \Gamma _{2}}}\left( \frac{\sigma _{2}^{2}%
}{2}\delta _{y}^{2}V_{i,j}^{n}\right) -\max \left\{ \overline{b_{12}}\delta
_{xy}^{+}V_{i,j}^{n},\ \underline{b_{12}}\delta
_{xy}^{-}V_{i,j}^{n} \right\}\notag\\
&&=\delta _{t}(u_{i,j}^{n}-U_{i,j}^{n})-\sup\limits_{_{\sigma _{1}^{2}\in
\Gamma _{1}}}\left( \frac{\sigma _{1}^{2}}{2}\delta
_{x}^{2}(u_{i,j}^{n}-U_{i,j}^{n})\right) -\sup\limits_{_{\sigma _{2}^{2}\in
\Gamma _{2}}}\left( \frac{\sigma _{2}^{2}}{2}\delta
_{y}^{2}(u_{i,j}^{n}-U_{i,j}^{n})\right) \notag\\
&&\ \  -\max \left\{ \overline{b_{12}}\delta
_{xy}^{+}(u_{i,j}^{n}- U_{i,j}^{n}),\underline{b_{12}}\delta
_{xy}^{-}(u_{i,j}^{n}- U_{i,j}^{n})\right\}  -\max\limits_{n} \|R^n_{up}\|_\infty \notag\\
&&\leq \delta _{t}u_{i,j}^{n}-\sup\limits_{_{\sigma _{1}^{2}\in \Gamma
_{1}}}\left( \frac{\sigma _{1}^{2}}{2}\delta _{x}^{2}u_{i,j}^{n}\right)
-\sup\limits_{_{\sigma _{2}^{2}\in \Gamma _{2}}}\left( \frac{\sigma _{2}^{2}%
}{2}\delta _{y}^{2}u_{i,j}^{n}\right) -\max \left\{ \overline{b_{12}}\delta
_{xy}^{+}u_{i,j}^{n},\underline{b_{12}}\delta
_{xy}^{-}u_{i,j}^{n}\right\} -\max\limits_{n} \|R^n_{up}\|_\infty\notag \\
&&\ \ -\left( \delta _{t}U_{i,j}^{n}-\sup\limits_{_{\sigma _{1}^{2}\in \Gamma
_{1}}}\left( \frac{\sigma _{1}^{2}}{2}\delta _{x}^{2}U_{i,j}^{n}\right)
-\sup\limits_{_{\sigma _{2}^{2}\in \Gamma _{2}}}\left( \frac{\sigma _{2}^{2}%
}{2}\delta _{y}^{2}U_{i,j}^{n}\right) -\max \left\{ \overline{b_{12}}\delta
_{xy}^{+} U_{i,j}^{n},\underline{b_{12}}\delta
_{xy}^{-}U_{i,j}^{n}\right\}\right) \notag \\
&&\leq R^n_{up}-\max\limits_{n} \|R^n_{up}\|_\infty  
 \leq 0,%
  \end{eqnarray}
where we have used \eqref{implicit  correlation2} and \eqref{u-up}.
Thanks to the initial condition $V_{i,j}^0 = 0$, and the boundary condition $V_{i,j}^n|_{(x_i,y_j)\in \partial \Omega} = - t^n*\max\limits_{n} \|R^n_{up}\|_\infty$, we have, by the maximum principle Proposition \eqref{EP}, that 
$$V^n_{i,j}\leq 0, \ \mbox{for}\ 0\leq i,j\leq M,\ 0\leq n\leq N.$$
That is 
\begin{equation}
u_{i,j}^{n}-U_{i,j}^{n}\leq t^{n}*\max\limits_{n} \|R^n_{up}\|_\infty.  \label{rate-up}
\end{equation}%

Similarly, let $W_{i,j}^{n}=U_{i,j}^{n}-u_{i,j}^{n}- t^{n}*\max\limits_n\|R^n_{low}\|_\infty$, then from \eqref{u-low}, we have 
\begin{equation}
U_{i,j}^{n}-u_{i,j}^{n}\leq  t^{n}*\max\limits_n\|R^n_{low}\|_\infty.  \label{rate-low}
\end{equation}%
Hence, we finally have 
$$\max_n \|u^n-U^n\|_\infty \leq T*\max_n(\|R^n_{up}\|_\infty+\|R^n_{low}\|_\infty)\leq C ((\Delta t)^{\beta/2} + h^\beta).$$
\end{proof}

\begin{remark} In Theorem \ref{rate}, we assume the regularity $u\in C^{1+\beta/2,
2+\beta} \left([0, T] \times \Omega\right)$, where $0<\epsilon <T.$ Actually, the viscosity solution $u$ of equation \eqref{UVCM}, has been proven to belong to the space $C^{1+\beta/2,
2+\beta} \left([\epsilon, T] \times \mathbb{R}^2\right)$, where $0<\epsilon <T$ and $\beta \in
(0,1)$\textsuperscript{\cite{p10}}. With consistent initial and boundary conditions, the solution of \eqref{general model} may be smooth up to the boundary of the
interval $[0, T]$.
\end{remark}


\section{Numerical examples}
\label{sec:5} 
In this section, we present some numerical examples to show the efficiency of our numerical scheme.  

\begin{example}
The following problem has an exact solution  $u=\sin (5(x+y+t))$.
\begin{equation*}
\left\{
\begin{array}{l}
u_{t}-\max\limits_{\substack{ \sigma _{1}^{2}\in \Gamma _{1},\sigma
_{2}^{2}\in \Gamma _{2},  \\ b_{12}\in \Gamma _{12}}}%
\left( \frac{\sigma _{1}^{2}}{2}u_{xx}+\frac{\sigma _{2}^{2}}{2}u_{yy}+b_{12}u_{xy}\right) =f, \\
u|_{t=0}=\sin (5(x+y)), \\
u|_{(x,y)\in \partial \Omega }=\sin (5(x+y+t))|_{(x,y)\in \partial \Omega },%
\end{array}%
\right.
\end{equation*}%
where $t\in (0,1),$ $\Omega =(-1,1)\times (-1,1),$ $\sigma
_{1}\in \lbrack 0.2,0.3],$ $\sigma _{2}\in \lbrack 0.25,0.35],$ $b_{12}\in \lbrack -0.04,0.03]$, and 
\begin{equation*}
f=5\cos (w)+25\min\limits_{\substack{ \sigma _{1}^{2}\in \Gamma _{1},\sigma
_{2}^{2}\in \Gamma _{2},  \\ b_{12}\in \Gamma _{12}}}%
\left( \frac{\sigma _{1}^{2}}{2}\sin (w)+\frac{\sigma _{2}^{2}}{2}\sin
(w)+b_{12}\sin (w)\right),
\end{equation*}%
with $w=5(x+y+t)$.
\label{exam2}
\end{example}

Due to the high regularity of the solutions of this equation, it is
theoretically not difficult to derive the following error estimation between
the numerical solution and the exact solution
\begin{equation*}
\left\vert \left\vert u^{n}-U^{n}\right\vert \right\vert _{\infty }\leq
O\left( \Delta t+h^{2}\right).
\end{equation*}
From Table \ref{ex2}, it can be observed that the error rate between the
numerical solution and the exact solution is first-order in time and
second-order in space, in terms of the norm of $\mathcal{L}%
^{\infty}(0,1;\Omega)$ and $\mathcal{L}^2(0,1;\Omega)$.
\begin{table}[!ht]
\caption{Error accuracy}
\label{ex2}\centering
{\footnotesize \centering
\begin{tabular}{@{}ccccccc}
\toprule Timesteps & Nodes & $\mathcal{L}^{\infty}(0,1;\Omega)$-error &
error order & $\mathcal{L}^2(0,1;\Omega)$-error & error order &  \\ \hline
50 & 11$\times$11 & 1.9013e-01 &  & 1.6062e-01 &  &  \\ \hline
200 & 21$\times$21 & 5.1659e-02 & 1.8799 & 4.3689e-02 & 1.8783 &  \\ \hline
800 & 41$\times$41 & 1.3075e-02 & 1.9822 & 1.1067e-02 & 1.9810 &  \\ \hline
3200 & 81$\times$81 & 3.2597e-03 & 2.0040 & 2.7594e-03 & 2.0038 &  \\
\toprule &  &  &  &  &  &
\end{tabular}
}
\end{table}

In Figure \ref{iterator1},  we present the number of inner iterations at each time step. It can be seen that the number of iterations per time step varies only from 3 to 5, demonstrating the high efficiency of the implicit numerical algorithm.
\begin{figure}[!h]
\center
\includegraphics[scale=0.45]{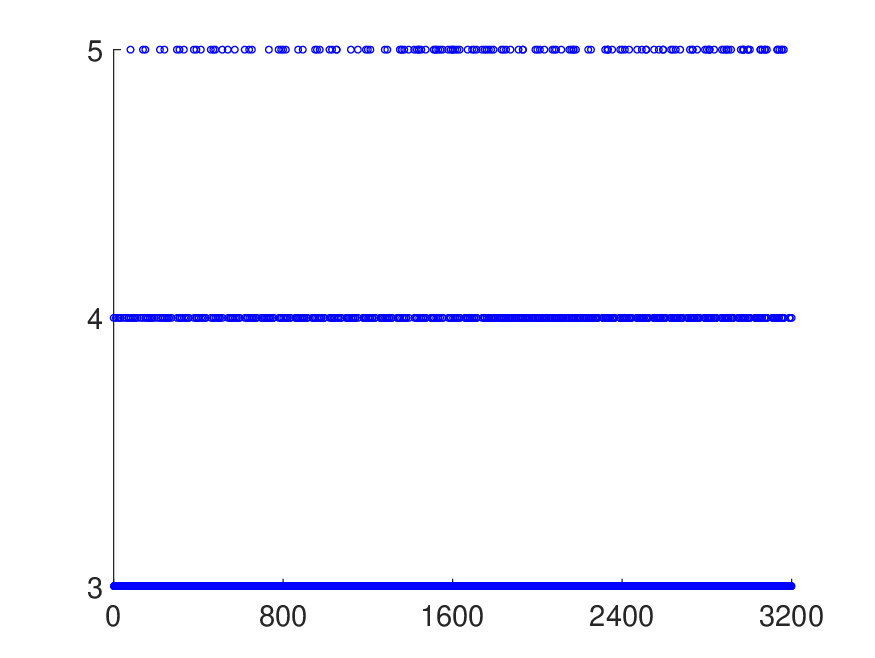} 
\caption{The number of iterations within each time step.}
\label{iterator1}
\end{figure}

\begin{example}
Consider a problem as follows
\begin{equation*}
\left\{
\begin{array}{l}
u_{t}-\max\limits_{\substack{ \sigma _{1}^{2}\in \Gamma _{1},\sigma
_{2}^{2}\in \Gamma _{2},  \\ b_{12}\in \Gamma _{12}}}%
\left( \frac{\sigma _{1}^{2}}{2}u_{xx}+\frac{\sigma _{2}^{2}}{2}u_{yy}+b_{12}u_{xy}\right) =0, \\
u|_{t=0}=\sin (5(x+y)), \\
u|_{(x,y)\in \partial \Omega }=\sin (5(x+y+t))|_{(x,y)\in \partial \Omega },%
\end{array}%
\right.
\end{equation*}%
where the parameter settings are identical to those in Example \ref{exam2}.
A reference \textquotedblleft exact" solution is token as the
numerical solution on a sufficiently fine grid (time step $\Delta t=1/5000$,
space step $h=1/180$). \label{exam4}
\end{example}

According to Table \ref{ex22}, it can be observed that the error order
between the numerical solution and the exact solution is approximately
first-order in time and second-order in space, in terms of the norms of $%
\mathcal{L}^{\infty }(0,1;\Omega )$ and $\mathcal{L}^{2}(0,1;\Omega )$. It
should be noted that the error order is higher than our theoretical results,
which is a very interesting phenomenon.
\begin{table}[!ht]
\caption{Error accuracy}
\label{ex22}
\centering
{\footnotesize \centering
\begin{tabular}{@{}ccccccc}
\toprule Timesteps & Nodes & $\mathcal{L}^{\infty}(0,1;\Omega)$-error &
error order & $\mathcal{L}^2(0,1;\Omega)$-error & error order &  \\ \hline
50 & 11$\times$11 & 2.3641e-01 &  & 1.4388e-01 &  &  \\ \hline
200 & 21$\times$21 & 9.0651e-02 & 1.3829 & 4.9501e-02 & 1.5393 &  \\ \hline
800 & 41$\times$41 & 2.8399e-02 & 1.6475 & 1.3452e-02 & 1.8796 &  \\ \hline
3200 & 81$\times$81 & 7.3077e-03 & 1.9584 & 3.3629e-03 & 2.000 &  \\
\toprule &  &  &  &  &  &
\end{tabular}
}
\end{table}

We also present the number of inner iterations in each time step. Figure \ref{iterator2} demonstrates that the number of iterations per time step typically ranges from 3 to 4, which also indicates the high efficiency of our implicit numerical algorithm. The
changes in volatilities $\sigma _{1}$, $\sigma _{2}$ and covariance $b_{12}$ over time are illustrated in Figures \ref{sigma11}, %
\ref{sigma22} and \ref{irho2}, respectively, on the mesh with 3200 time steps and 81$\times $81 spatial grid points. 
\begin{figure}[!ht]
\center
\includegraphics[scale=0.45]{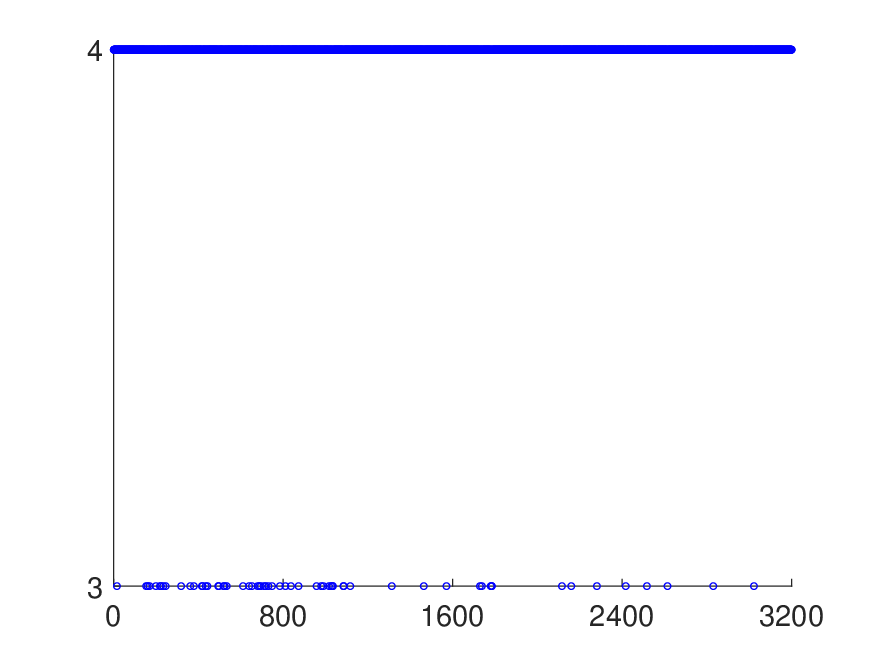} 
\caption{The number of iterations within each time step.}
\label{iterator2}
\end{figure}

\begin{figure}[!ht]
\center
\includegraphics[scale=0.82]{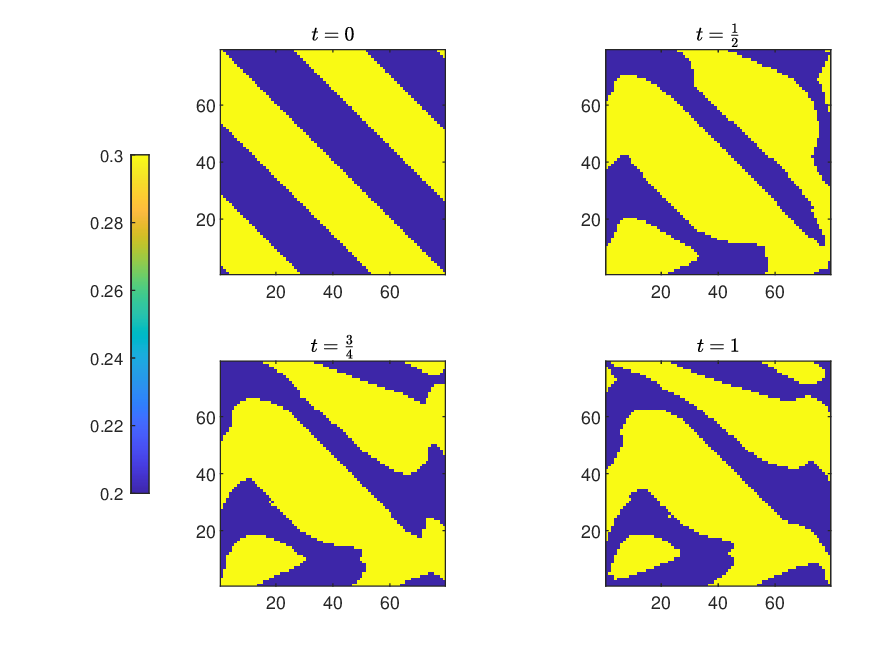}
\caption{The graph of the variation of $\protect\sigma _{1}$ over time.}
\label{sigma11}
\end{figure}

\begin{figure}[!ht]
\center
\includegraphics[scale=0.82]{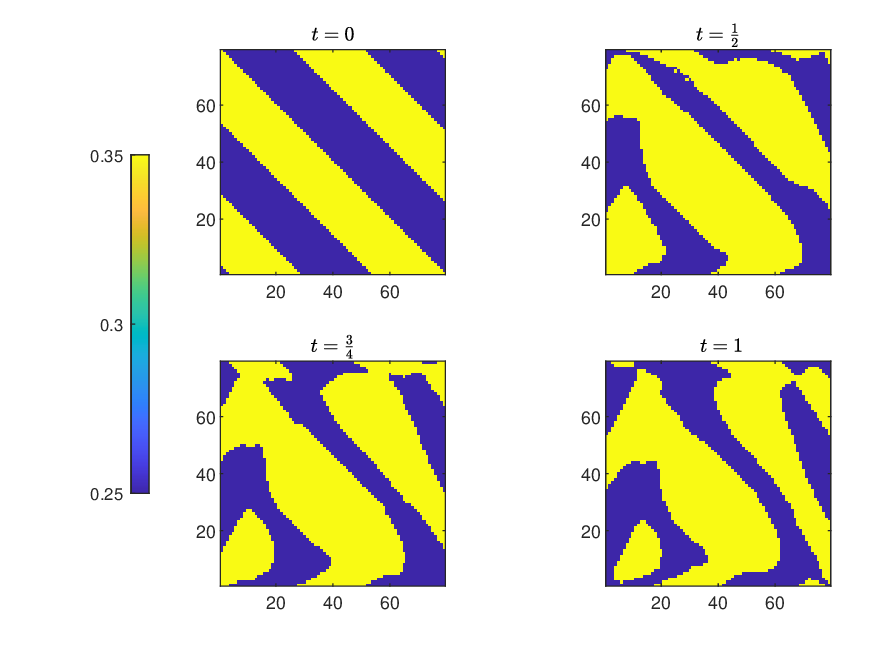} 
\caption{The graph of the variation of $\protect\sigma_{2}$ over time.}
\label{sigma22}
\end{figure}

\begin{figure}[!ht]
\center
\includegraphics[scale=0.82]{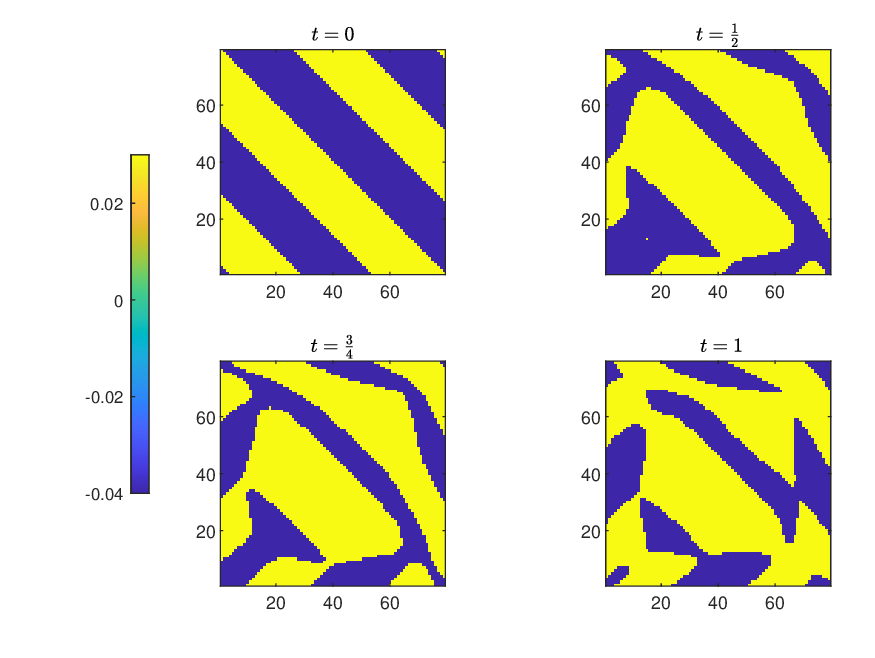} 
\caption{The graph of the variation of $b_{12}$ over time.}
\label{irho2}
\end{figure}

\section{Conclusions}
\label{section4} We have developed an implicit discretization scheme
to solve the general two-dimensional G-heat equation, which particularly addresses cases where the sign of the correlationship is uncertain. To ensure convergence of the implicit discretization scheme, we require
that the covariance matrix of the two underlying assets be diagonally
dominant, which is not too restrictive. First, we prove the monotonic
convergence of the non-linear inner iteration at each time step. Then, we
demonstrate the consistency, stability, and monotonicity of the numerical
scheme, thus establishing its convergence. Furthermore, we provide an
estimate of the convergence order. Finally, we provide corresponding
numerical examples and find that although the implicit numerical algorithm
involves inner iterations at each time step, it remains highly efficient
with a computational complexity about 3-4 times that of solving linear
expectations.

\section*{Acknowledgments}
We are grateful to Professors Shige Peng and Lihe Wang for useful discussions. This work of Yue was partially supported by the NSFC Grant 12371401.


\end{document}